\documentclass[12pt,journal,draftcls,onecolumn]{IEEEtran}
\usepackage{amsmath,amssymb,mathrsfs}
\usepackage{epsfig,epsf,subfigure,graphicx,graphics}
\usepackage{url}

\newtheorem{lemma}{Lemma}[section]
\newtheorem{theorem}[lemma]{Theorem}

\newtheorem{definition}[lemma]{Definition}
\newtheorem{remark}[lemma]{Remark}

\newcommand{\SNR}{\mathsf{SNR}}
\newcommand{\INR}{\mathsf{INR}}

\newcommand{\Var}{\mathrm{Var}}
\newcommand{\Cov}{\mathrm{Cov}}

\newcommand{\etal}{{\it et al.}}

\newcommand{\aaaa}{\mathrm{(a)}}
\newcommand{\bbbb}{\mathrm{(b)}}
\newcommand{\cccc}{\mathrm{(c)}}
\newcommand{\dddd}{\mathrm{(d)}}
\newcommand{\eeee}{\mathrm{(e)}}
\newcommand{\ffff}{\mathrm{(f)}}

\newcommand{\lp}{\left(}
\newcommand{\rp}{\right)}
\newcommand{\lb}{\left[}
\newcommand{\rb}{\right]}
\newcommand{\lbp}{\left\{}
\newcommand{\rbp}{\right\}}
\newcommand{\ul}{\underline}
\newcommand{\ol}{\overline}
\newcommand{\mcal}{\mathcal}

\newcommand{\msf}{\mathsf}

\newcommand{\mb}{\mathbf}
\newcommand{\mbb}{\mathbb}

\newcommand{\ra}{\rightarrow}

\graphicspath{{fig/}}

\allowdisplaybreaks

\title{Distributed Interference Cancellation in Multiple Access Channels}
\author{
\authorblockN{I-Hsiang Wang}\\
\authorblockA{Wireless Foundations\\
University of California at Berkeley,\\
Berkeley, California 94720, USA\\
\textsf{ihsiang@eecs.berkeley.edu}}
}

\date{}

\begin{document}
\maketitle

\begin{abstract}
%In this paper, we consider a Gaussian multiple access channel with multiple independent additive white Gaussian interferences. Each interference is known to exactly one transmitter non-causally. The capacity region is characterized to within a constant gap, regardless of channel parameters. In the case of two users, we characterize the capacity region to within $1$ and $0.5$ bits for the stronger user and the weaker user respectively. In the case where transmitters are allowed to cooperate through finite-capacity links. The capacity region is characterized to within $3$ and $1.5$ bits for the stronger user and the weaker user respectively. These results are based on a layered modulo-lattice transmission architecture which realizes distributed interference cancellation.
In this paper, we consider a Gaussian multiple access channel with multiple independent additive white Gaussian interferences. Each interference is known to exactly one transmitter non-causally. 
%We make use of a binary expansion model to help uncover the underlying layered structure of the original channel, whic
The capacity region is characterized to within a constant gap regardless of channel parameters. These results are based on a layered modulo-lattice scheme which realizes distributed interference cancellation.
\end{abstract}


\section{Introduction}
In modern wireless communication systems, interference has become the major barrier for efficient utilization of available spectrum. 
%In many scenarios, including inter-symbol interference in frequency selective channels, inter-carrier interference in FDM systems, etc., interference is originated from the transmitter itself and hence known to the transmitter.
%In many scenarios, interference can be inferred by intelligent transmitters thanks to various estimation and statistical learning techniques, while receivers cannot due to physical and/or power limitations. 
In many scenarios, interferences are originated from sources close to transmitters and hence can be inferred by intelligent transmitters, while receivers cannot due to physical limitations. 
With the knowledge of interferences as side information, transmitters are able to encode their information against interferences and mitigate them, even though receivers cannot distinguish interferences from desired signals. The simplest information theoretic model for studying such interference mitigation is the single-user point-to-point dirty-paper channel \cite{Costa_83}, which is a special case of state-dependent memoryless channels with the state\footnote{In dirty-paper channel, the state is the additive interference.} known non-causaully to the transmitter \cite{GelfandPinsker_80}. It is shown that the effect of interference can be completely removed in the additive white Gaussian noise (AWGN) channel when the interference is also additive white Gaussian \cite{Costa_83}. 
%A comprehensive survey on channel coding with side information can be found in \cite{KeshetSteinberg_08}.
As for multi-user scenarios, it has been found that when perfect state information (the additive interference) is available non-causally at all transmitters, the capacity region of the AWGN multiple access channel (MAC) is not affected by the additive white Gaussian interference \cite{GelfandPinsker_84} \cite{KimSutivong_04}. When the sate information is known \emph{partially} to different transmitters in the MAC, however, the capacity loss caused by the interference is unbounded as the signal-to-noise ratios increase \cite{PhilosofZamir_07} \cite{Somekh-BaruchShamai_08}. Since each transmitter only has partial knowledge about the interference, interference cancellation has to be realized in a \emph{distributed} manner.

In this paper, we consider a $K$-user Gaussian MAC with $K$ independent additive white Gaussian interferences. Each interference is known to exactly one transmitter non-causally. Transmitter $i$, for all $i=1,\ldots,K$, aims to deliver a message $w_i$ to the receiver reliably through the channel depicted in Fig.~\ref{fig_Model}, where
%Furthermore, we allow transmitters to cooperate through finite-capacity links, so that transmitters can cooperatively transmit their messages and/or mitigate the known interferences. This is exactly the same model studied in \cite{PhilosofZamir_07} except for the transmitter cooperation. For simplicity, we mainly focus on the two-user case, termed as \emph{doubly-dirty MAC} in \cite{PhilosofZamir_07}. 
%The model is depicted in Fig.~\ref{fig_Model}, where
\begin{align*}
%y = x_1+x_2+s_1+s_2+z,
y = \sum_{i=1}^K x_i + \sum_{i=1}^K s_i + z,
\end{align*}
and $z\sim\mcal{N}\lp 0, N_o\rp$ is the AWGN noise. Interference $s_i\sim\mcal{N}\lp 0,Q_i\rp$, $i=1,\ldots,K$, independent of everything else, is known non-causally to transmitter $i$ \emph{only}. Power constraint at transmitter $i$ is $P_i$, $i=1,\ldots,K$. Define channel parameters $\SNR_i := P_i / N_o$, $\INR_i := Q_i / N_o$, for $i=1,2$. 
%Transmitter cooperation is induced by two orthogonal noise-free links with capacity $\C_{12}$ and $\C_{21}$, which carry $t_{12}$ and $t_{21}$ respectively. 
User $i$'s rate is denoted by $R_i$, $i=1,\ldots,K$. Throughout this paper, without loss of generality we assume that $P_1\ge P_2 \ge \ldots \ge P_K$.

%In the case of $K=2$, the channel is termed as \emph{doubly-dirty MAC} in \cite{PhilosofZamir_07}. We further allow transmitters to cooperate through finite-capacity links, so that transmitters can cooperatively transmit their messages and/or mitigate the known interferences. This is exactly the same model studied in \cite{PhilosofZamir_07} except for the transmitter cooperation. The model is depicted in Fig.~\ref{fig_Model} (b), where the transmitter cooperation is induced by two orthogonal noise-free links with capacity $\C_{12}$ and $\C_{21}$, which carry $t_{12}$ and $t_{21}$ respectively.

\begin{figure}[htbp]
{\center
\includegraphics[height=1.5in]{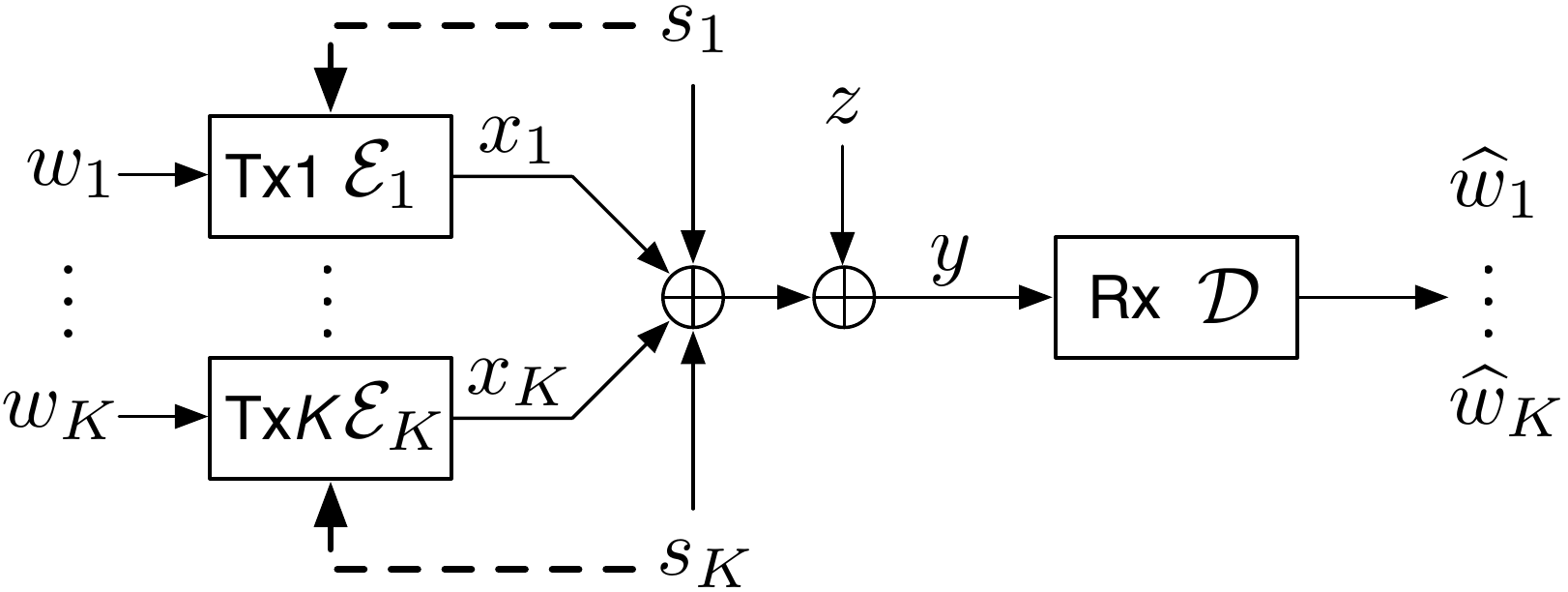}
%\subfigure[$K$-user MAC without Cooperation]{\includegraphics[height=1.5in]{Gaussian_KK.pdf}}
%\subfigure[Doubly Dirty MAC with Cooperation]{\includegraphics[height=1.5in]{Gaussian_coop.pdf}}
\caption{Channel Model}
\label{fig_Model}
}
\end{figure}

\subsection{Related Works}
State-dependent networks with partial state knowledge available at different nodes have been studied in various scenarios. Kotagiri \etal \cite{KotagiriLaneman_08} study the state-dependent two-user MAC with state non-causally known to a transmitter, and for the Gaussian case they characterize the capacity asymptotically at infinite interference ($K=2, Q_1=\infty, Q_2=0$) as the informed transmitter's power grows to infinity. Somekh-Baruch \etal \cite{Somekh-BaruchShamai_08} study the problem with the same set-up as \cite{KotagiriLaneman_08} while the informed transmitter knows the other's message, and they characterize the capacity region completely. Zaidi \etal \cite{ZaidiKotagiri_09} study another case of degraded message set. The achievability part of \cite{KotagiriLaneman_08}, \cite{Somekh-BaruchShamai_08}, and \cite{ZaidiKotagiri_09} are based on random binning. Philosof \etal \cite{PhilosofZamir_07}, on the other hand, characterize the capacity region of the doubly-dirty MAC to within a constant gap at infinite interference (i.e., $K=2$, $Q_1=Q_2=\infty$), by lattice strategies \cite{ErezShamai_05}. They also show that strategies based on Gaussian random binning is unboundedly worse than lattice-based strategies. Zaidi \etal \cite{ZaidiKotagiri_10} \cite{ZaidiShamai_10} and Akhbari \etal \cite{AkhbariMirmohseni_09} study a state-dependent relay channel where the state is only known either at the source or the relay.

%On the other hand, transmitter cooperation has also been widely investigated in various scenarios, and a non-exhaustive list includes MAC \cite{Willems_83} \cite{BrossLapidoth_08}, interference channels \cite{PrabhakaranViswanathSRC_09} \cite{WangTse_10}, MAC with state known at all transmitters \cite{BrossLapidoth_08}, and MAC with partial state known at transmitters and full state knowledge at the receiver \cite{PermuterShamai_10}.

\subsection{Main Contribution}\label{subsec_Main}
%We characterize the capacity region of the doubly-dirty MAC without cooperation to within $1$ and $0.5$ bits for $R_1$ and $R_2$ respectively, and hence extend the constant-gap-to-optimality result in \cite{PhilosofZamir_07} to \emph{arbitrary} interference powers. 
We characterize the capacity region of the channel in Fig.~\ref{fig_Model} to within $K\log_2 K$ bits, regardless of channel parameters $P_i$'s, $Q_i$'s, and $N_o$. The constant gap only depends on the number of users in the channel and is independent of channel parameters, providing a strong guarantee on the performance for any fixed $K$. Our approach to this problem is first investigating a \emph{binary expansion model} of the original channel. The binary expansion model is a natural extension of the linear deterministic model proposed in \cite{AvestimehrDiggavi_07} to the case with additive interferences known to transmitters. After characterizing the capacity region of the binary expansion model, we then make use of the intuitions and techniques developed there to derive outer bounds and build up achievability results for the original Gaussian problem. Such approach has been successfully applied to various problems in network information theory, including \cite{BreslerParekh_08}, \cite{AvestimehrDiggavi_09}, \cite{RiniTuninetti_10}, \cite{WangTse_09}, \cite{WangTse_10}, etc. 
%It turns out that in our problem, the capacity regions of the binary expansion model and the original Gaussian model are within a constant number of bits.

For the achievability part we propose a layered modulo-lattice scheme consisting of $K$ layers, based on the intuition drawn from the study of the binary expansion model. Layer $i$ is shared among user $1,\ldots,i$, and the hierarchy of the layers is $1\ra2\ra\ldots\ra K$, from the top to the bottom. Each layer treats the signals sent at higher layers as \emph{interference}, each of which is known non-causally to exactly one transmitter. In each layer $i\in1,2\ldots,K$, we use a modulo-lattice scheme to realize distributed interference cancellation, which is a simpler version of the single layer scheme in \cite{PhilosofZamir_07}.
%similar to the single layer scheme in \cite{PhilosofZamir_07}. 
For the converse part, we first extend the ideas in \cite{PhilosofZamir_09} to derive matching outer bounds for the binary expansion model and then use the same technique to prove bounds in the Gaussian scenario.

%In the two-user case with transmitter cooperation, for the achievability part we introduce layer $\mfrak{C}$ in addition to layer $1$ and $2$, and the hierarchy of the layers is $1\ra \mfrak{C}\ra 2$, from the top to the bottom. In layer $\mfrak{C}$, the weaker transmitter Tx2 compresses the precoded information (precoded against intereference $s_2$) at a proper distortion, and uses part of the cooperation capacity to send the compression index to Tx1. Then Tx1 precodes the compressed signal along with user 1's information against the aggregate interference at this layer. In effect, only Tx1 sends out a signal in this layer.
%In layer $\mfrak{R}$, Tx2 uses the rest of the cooperation capacity to send additional data to Tx1. Tx1 uses the rest of its power to further transmit its own information or relay user 2's information, precoded against $s_1$ using either Gaussian random binning \cite{Costa_83} or lattice strategies \cite{ErezShamai_05}\footnote{The name ``Gaussian random binning layer $\mfrak{R}$" is to stress that Gaussian random binning and lattice strategies are equally good.}. For the outer bound, we use a similar argument as \cite{PhilosofZamir_09}.

%Due to space constraints, converse proofs and gap analysis are omitted in this paper. Please refer to the extended version \cite{Wang_10} for details.

\subsection{Notations}
Notations used in this paper are summarized below:
\begin{itemize}
\item
Throughout the paper, the block coding length is denoted by $N$. A sequence of random variables $x[1],\ldots, x[N]$ is denoted by $x^N$ and boldface $\mb{x}$ interchangeably. 
\item
Logarithms are of base $2$ if not specified. We use short-hand notations 
%$\mcal{C}\lp \cdot \rp$ to denote $\frac{1}{2}\log\lp 1+ \cdot \rp$, 
$\lp\cdot\rp^+$ to denote $\max\lbp0, \cdot\rbp$ and $\log^+\lp\cdot\rp$ to denote $\lp\log\lp\cdot\rp\rp^+$. 
%\item
%$\mbb{I}\lbp A\rbp$ denotes the indicator function, evaluated to $1$ if event $A$ is true and $0$ otherwise.
\item
We use the short-hand notation $[k_1:k_2]$ to denote a set/tuple $\lp k_1, \ldots, k_2\rp$ and $v_{[k_1:k_2]}$ to denote $\lp v_{k_1}, \ldots, v_{k_2}\rp$ if $k_1\le k_2$, respectively. If $k_1 > k_2$, $[k_1:k_2]$ and $v_{[k_1:k_2]}$ denote the empty set $\phi$.
\item
Similarly, for a set of indices $S$, we use $v_{S}$ to denote the collection $\lbp v_i|\ i\in S\rbp$.
\end{itemize}

\subsection{Paper Organization}
The rest of this paper is organized as follows. In Section~\ref{sec_LDC}, we first introduce and formalize the binary expansion model, which serves as an auxiliary channel for the original one. Then we characterize the capacity region of the auxiliary channel and draw important intuitions for solving the original problem. In Section~\ref{sec_Lattice}, we propose the layered modulo-lattice scheme and derive its achievable rates. Then we show that the achievable rate region is within a constant gap to the proposed outer bounds in Section~\ref{sec_ConstGap}. Finally, we conclude the paper in Section~\ref{sec_Conclude}.
% with some discussions on extensions of this work to more general multi-user networks.

\section{A Binary Expansion Model for Gaussian MAC with Additive Interferences}\label{sec_LDC}
To approach the distributed interference cancellation problem in Gaussian multiple access channels (MAC), we first study a binary expansion model of the original problem. Solutions to the original problem can be inferred by solving the auxiliary problem in this model. The model is a natural generalization of the linear deterministic model proposed in \cite{AvestimehrDiggavi_07}, with random states acting as additive interferences. We formally define the model as follows.

\begin{definition}[Binary Expansion Model]\label{def_LDC}
The binary expansion MAC with additive interferences known to transmitters, corresponding to the original Gaussian problem, is defined by nonnegative integers 
\begin{align}
n_i := \left\lfloor \frac{1}{2}\log^+ \SNR_i \right\rfloor,\ m_i := \left\lfloor \frac{1}{2}\log^+ \INR_i \right\rfloor,\ i=1,\ldots,K,
\end{align}
transmitted signals $x_{b,i} \in \mbb{F}_2^q$, interferences $s_{b,i} \in \mbb{F}_2^q$ for $i\in[1:K]$, and received signal
\begin{align}
y_b = \sum_{i=1}^K  A^{q-n_i}x_{b,i} + \sum_{i=1}^K A^{q-m_i}s_{b,i},
\end{align}
where additions are modulo-two component-wise, $q=\max\lbp n_i,m_i: i\in[1:K]\rbp$, and $A\in\mathbb{F}_2^{q\times q}$ is the shift matrix
\begin{align}
A = \lb \begin{array}{ccccc} 0 & 0 & 0 & \cdots & 0\\ 1 & 0 & 0 & \cdots & 0\\ 0 & 1 & 0 & \cdots & 0\\ \vdots & & & \ddots & \vdots\\ 0 & \cdots & 0 & 1 & 0\end{array}\rb.
\end{align}
Each interference $s_{b,i}$ consists of $q$ i.i.d. $\mathrm{Bernoulli}\lp\frac{1}{2}\rp$ bits and is known to transmitter $i$, for $i\in[1:K]$.
\end{definition}

Here we use subscript $b$ to draw distinction from the original channel model. 
%To simplify notations, however, in this section we drop the subscript $b$ when there is no potential ambiguity. 
Note that the condition $P_1\ge P_2 \ge \ldots \ge P_K$ implies $n_1\ge n_2 \ge \ldots \ge n_K$.

An example is depicted in Fig.~\ref{fig_LDC}, where $n_1=4,n_2=2,m_1=5,m_2=3$.
\begin{figure}[htbp]
{\center
\includegraphics[width=3in]{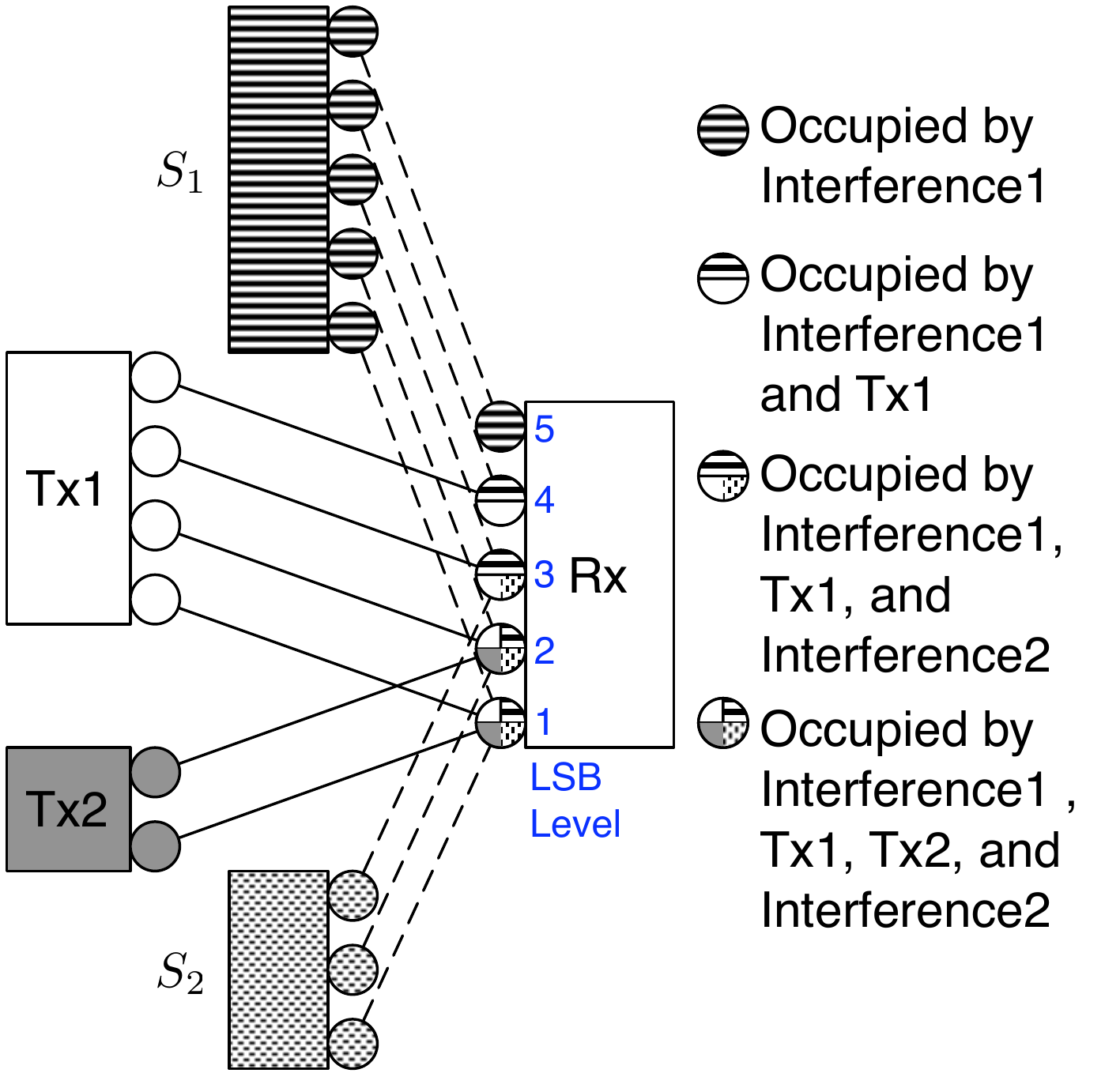}
\caption{The Binary Expansion Model. The numbering in blue denotes the ordering of the least significant bit (LSB) levels.}
\label{fig_LDC}
}
\end{figure}

The main result in this section is the characterization of capacity region of the auxiliary channel, summarized in the following theorem and two lemmas. To distinguish notations from the original Gaussian problem, lower-cases letters are used to represent rates in the binary expansion model. 

\begin{lemma}[Outer Bounds]\label{lem_LDCconverse} 
If $r_{[1:K]}\ge 0$ is achievable, it satisfies the following: for all $k \in [1:K]$,
\begin{align}
%\sum_{i=1}^k R_i \le \ol{\msf{R}}_k\lp n_{[1:k]}, m_{[1:k-1]}\rp,
\sum_{i=k}^K r_i \le \ol{\msf{r}}_k\lp n_{[k:K]}, m_{[k+1:K]};K\rp,
\end{align}
where 
\begin{align}
%\ol{\msf{R}}_k\lp n_{[1:k]}, m_{[1:k-1]}\rp := \max\lbp m_{[1:k-1]},n_k\rbp - \sum_{i=1}^{k-1}\lp m_i - \max\lbp m_{[1:i-1]}, n_i\rbp\rp^+.
\ol{\msf{r}}_k\lp n_{[k:K]}, m_{[k+1:K]};K\rp := \max\lbp m_{[k+1:K]},n_k\rbp - \sum_{i=k+1}^{K}\lp m_i - \max\lbp m_{[i+1:K]}, n_i\rbp\rp^+.
\end{align}
\end{lemma}
\begin{proof}
The proof is detailed in Section~\ref{subsec_LDCconverse}.
\end{proof}

\begin{lemma}[Inner Bounds]\label{lem_LDCachieve}
If $r_{[1:K]}\ge 0$ satisfies the following: for all $k \in [1:K]$,
\begin{align}
%\sum_{i=1}^k R_i \le \ul{\msf{R}}_k\lp n_{[1:k]}, m_{[1:k-1]}\rp,
\sum_{i=k}^K r_i \le \ul{\msf{r}}_k\lp n_{[k:K]}, m_{[k+1:K]};K\rp
\end{align}
it is achievable. Here
\begin{align}
%\ul{\msf{R}}_k\lp n_{[1:k]}, m_{[1:k-1]}\rp := \sum_{i=1}^k \lp n_i - \max\lbp m_{[1:i-1]}, n_{i-1}\rbp\rp^+.
\ul{\msf{r}}_k\lp n_{[k:K]}, m_{[k+1:K]};K\rp := \sum_{i=k}^K \lp n_i - \max\lbp m_{[i+1:K]}, n_{i+1}\rbp\rp^+.
\end{align}
\end{lemma}
\begin{proof}
The proof is detailed in Section~\ref{subsec_LDCAchieve}.
\end{proof}

\begin{theorem}[Capacity of the Binary Expansion Model]\label{thm_LDC} 
$r_{[1:K]}\ge 0$ is achievable, if and only if it satisfies the following: for all $k \in [1:K]$,
\begin{align}
%\sum_{i=1}^k R_i \le \msf{R}_k\lp n_{[1:k]}, m_{[1:k-1]}\rp,
\sum_{i=k}^K r_i \le \msf{r}_k\lp n_{[k:K]}, m_{[k+1:K]};K\rp,
\end{align}
%where $\msf{R}_k\lp n_{[1:k]}, m_{[1:k-1]}\rp = \ol{\msf{R}}_k\lp n_{[1:k]}, m_{[1:k-1]}\rp = \ul{\msf{R}}_k\lp n_{[1:k]}, m_{[1:k-1]}\rp$.
where $\msf{r}_k\lp n_{[k:K]}, m_{[k+1:K]};K\rp = \ol{\msf{r}}_k\lp n_{[k:K]}, m_{[k+1:K]};K\rp = \ul{\msf{r}}_k\lp n_{[k:K]}, m_{[k+1:K]};K\rp$.
\end{theorem}
\begin{proof}
To show $\ol{\msf{r}}_k\lp n_{[k:K]}, m_{[k+1:K]};K\rp = \ul{\msf{r}}_k\lp n_{[k:K]}, m_{[k+1:K]};K\rp$ for all $k\in[1:K]$, we shall use induction backwards.

{\flushleft 1) $k=K$}: $\ol{\msf{r}}_K\lp n_K;K\rp = n_K = \ul{\msf{r}}_K\lp n_K;K\rp$.
{\flushleft 2)} Suppose the claim is correct for $k=l$. For $k=l-1$,
\begin{align}
&\ol{\msf{r}}_{l-1}\lp n_{[l-1:K]}, m_{[l:K]};K\rp - \ol{\msf{r}}_l\lp n_{[l:K]}, m_{[l+1:K]};K\rp\\
&= \max\lbp m_{[l:K]},n_{l-1}\rbp - \max\lbp m_{[l+1:K]},n_{l}\rbp - \lp m_l - \max\lbp m_{[l+1:K]}, n_l\rbp\rp^+\\
&= \max\lbp m_{[l:K]},n_{l-1}\rbp - \max\lbp m_{[l+1:K]},n_{l},m_l\rbp\\
&= \max\lbp m_{[l:K]},n_l,n_{l-1}\rbp - \max\lbp m_{[l:K]},n_{l}\rbp\\
&= \max\lbp \max\lbp m_{[l:K]},n_l\rbp,n_{l-1}\rbp - \max\lbp m_{[l:K]},n_{l}\rbp\\
&= \lp n_{l-1} - \max\lbp m_{[l:K]}, n_{l}\rbp\rp^+ = \ul{\msf{r}}_{l-1}\lp n_{[l-1:K]}, m_{[l:K]};K\rp - \ul{\msf{r}}_l\lp n_{[l:K]}, m_{[l+1:K]};K\rp.
\end{align}
Hence, $\ol{\msf{r}}_{l-1}\lp n_{[l-1:K]}, m_{[l:K]};K\rp = \ul{\msf{r}}_{l-1}\lp n_{[l-1:K]}, m_{[l:K]};K\rp$. By induction principle, the proof is complete.
\end{proof}

\subsection{Motivating Examples}
Before we formally prove the converse and the achievability, we first give a couple of examples to illustrate the high-level intuition behind the result. Such intuitions not only work for the binary expansion model, but also carry over to the original Gaussian setting. For simplicity, all examples are two-user ($K=2$), with fixed $(n_1,n_2) = (4,2)$ and various $(m_1,m_2)$. They are depicted in Fig.~\ref{fig_Examples}. Although the total number of bit levels of $y_{b}$ is $q=\max\{ n_1, n_2, m_1,m_2\}$, referring to Fig.~\ref{fig_LDC} only the first $\max\{ n_1,n_2\}=4$ least significant bit (LSB) levels are those can be potentially used for communicating information, since none of the transmitters can access the upper bit levels owing to power constraints. Therefore, with the side information of interferences at transmitters, they try to cancel interferences in these $4$ levels as much as possible.

\begin{figure}[htbp]
{\center
\subfigure[$m_1\le n_1, m_2\le n_2$]{\includegraphics[width=2in]{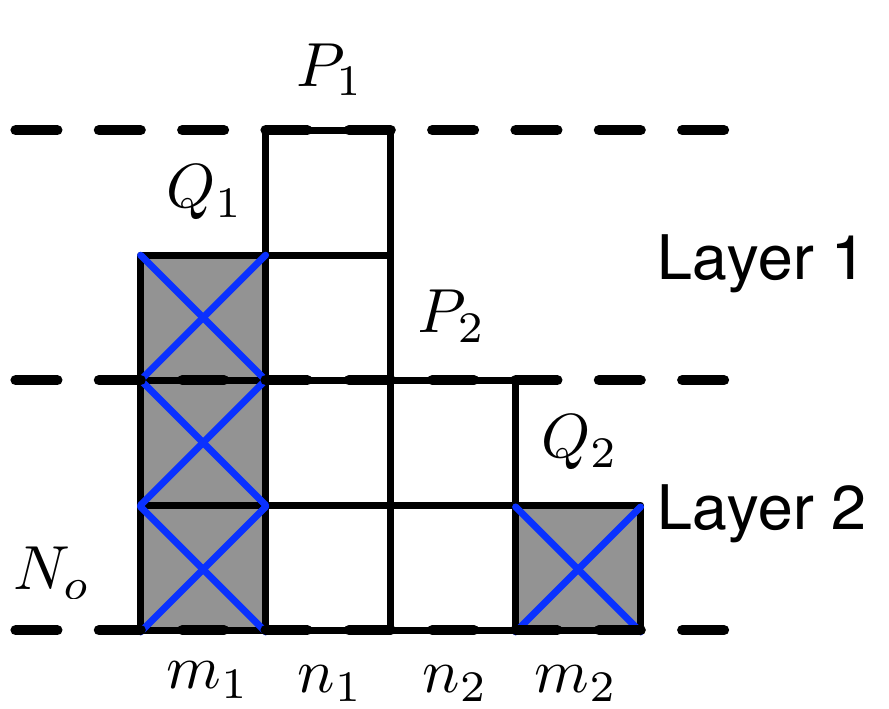}}
\subfigure[$m_1\ge n_1, m_2\le n_2$]{\includegraphics[width=2in]{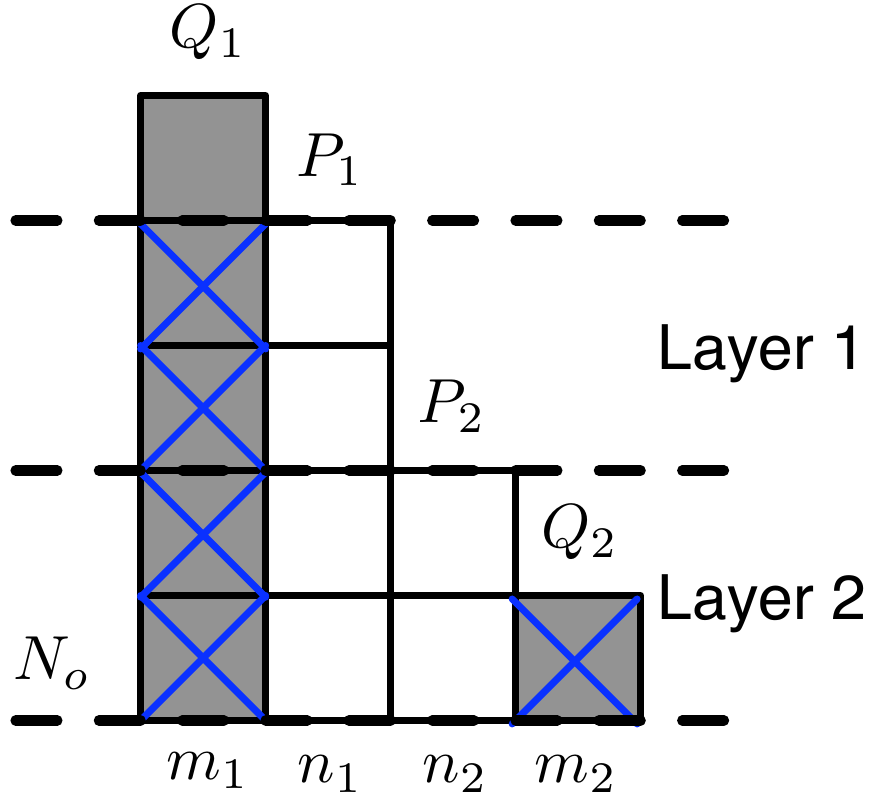}}
\subfigure[$n_2 \le m_2\le n_1$]{\includegraphics[width=2in]{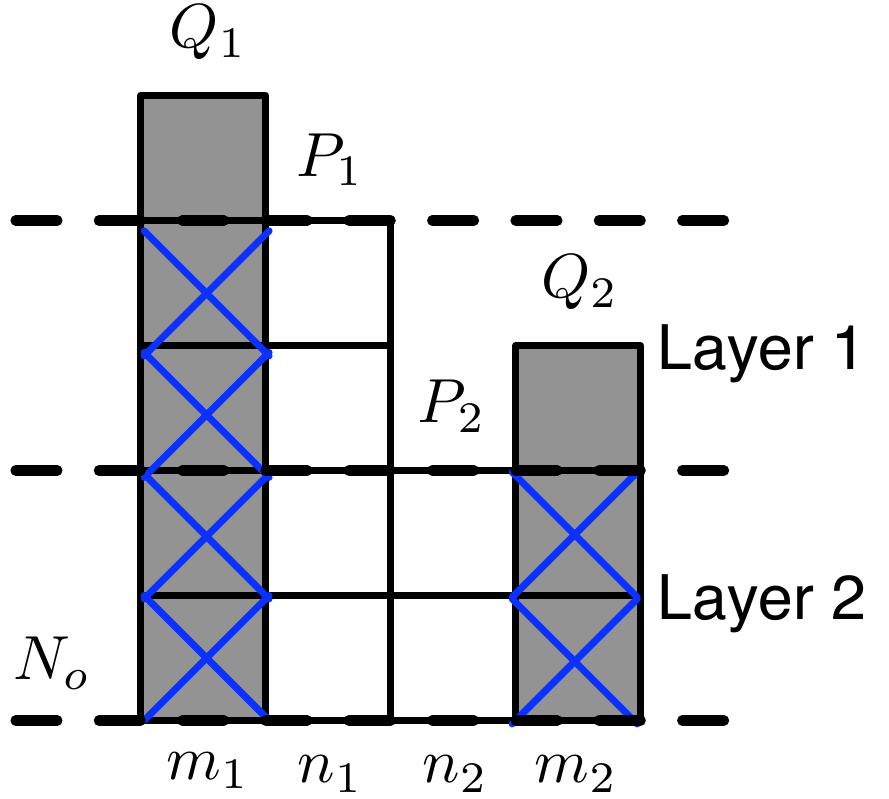}}
\subfigure[$m_2\ge n_1$]{\includegraphics[width=2in]{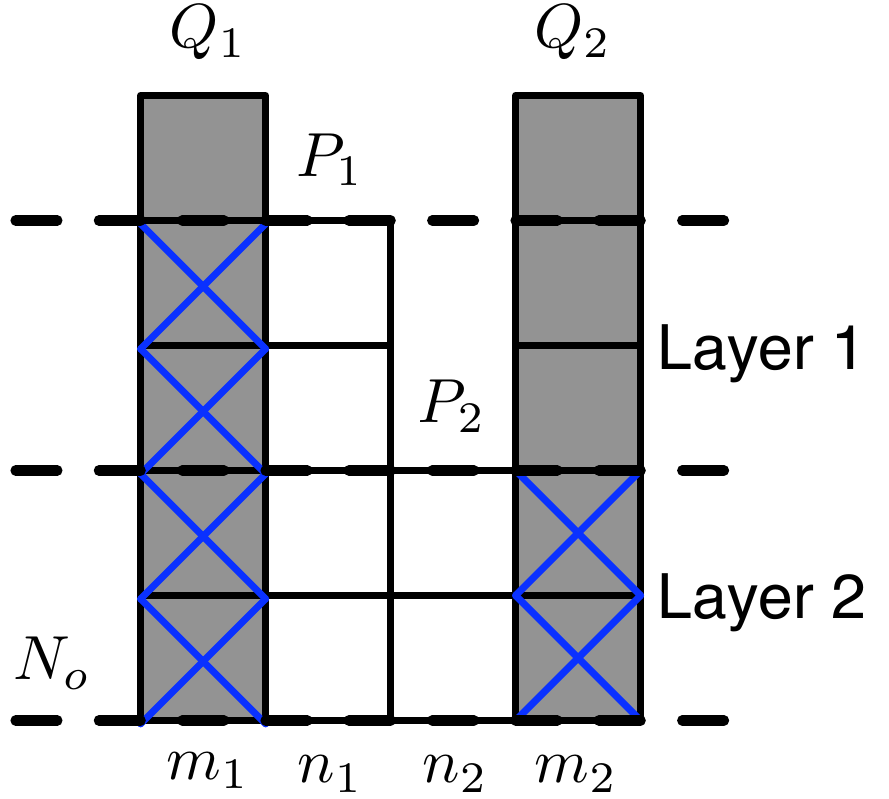}}
\subfigure[Capacity Region Degradation]{\includegraphics[width=4in]{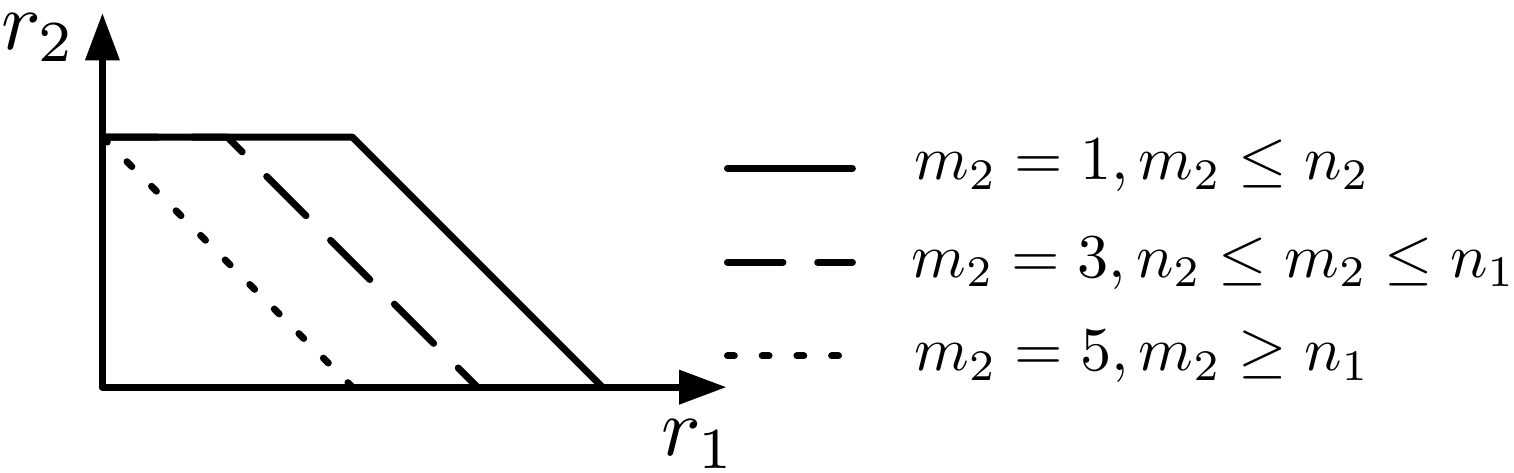}}
\caption{Examples. Note that by definition, $n_i \leftrightarrow P_i$, $m_i \leftrightarrow Q_i$, and the bottom line corresponds to the noise level. Blue crosses denote that interference bits (shaded) can be cancelled.}
\label{fig_Examples}
}
\end{figure}

The first example (Fig.~\ref{fig_Examples}(a)) illustrates the situation where $m_1\le n_1$ and $m_2\le n_2$. Transmitter 1 can completely cancel interference $s_{b,1}$ since it only occupies $m_1=3$ LSB levels of $y_{b}$. Transmitter 2 can also cancel interference $s_{b,2}$ completely since it only occupies $m_2=1$ LSB level of $y_{b}$. Therefore, all bit levels are free from interference, and the capacity region is $r_1+r_2\le 4, r_2\le 2$ which is the same as the clean MAC.

The second example (Fig.~\ref{fig_Examples}(b)) illustrates the situation where $m_1\ge n_1$ and $m_2\le n_2$. Transmitter 1 cannot completely cancel interference $s_{b,1}$ since it occupies $m_1=5$ LSB levels of $y_{b}$, while transmitter 1 has access to only $n_1=4$ LSB levels. However, it can cancel those in the first $4$ LSB levels. Transmitter 2 can again cancel interference $s_{b,2}$ completely. Therefore, all $4$ LSB levels are free from interference, and the capacity region is $r_1+r_2\le 4, r_2\le 2$ which is again the same as the clean MAC.

From the above examples, we see that the strength of interference $s_{b,1}$ does not effect the capacity region, since the only bit levels that matter are the $\max\{n_1,n_2\} = n_1 = 4$ LSB levels, and transmitter 1 can always ``clean up" the interference caused by $s_{b,1}$ in these bit levels. On the other hand, the strength of interference $s_{b,2}$ \emph{does} affect the capacity region, as discussed below.

The third example (Fig.~\ref{fig_Examples}(c)) illustrates the situation where $n_2\le m_2\le n_1$. Since transmitter 2 only has access to $n_2=2$ LSB levels, it cannot cancel the interference caused by $s_{b,2}$ at the third LSB level. Therefore, the level is no longer useful and cannot be used by transmitter 1. The fourth LSB level, however, is clean after transmitter 1's interference cancellation. The capacity region becomes $r_1+r_2 \le 3, r_2\le 2$. 

The last example (Fig.~\ref{fig_Examples}(d)) illustrates the situation where $m_2\ge n_1$. Again transmitter 2 cannot do anything about $s_{b,2}$ except at the $2$ LSB levels. Therefore, the third and the fourth bit levels are both corrupted and cannot be used. The capacity region becomes $r_1+r_2 \le 2$. 

Fig.~\ref{fig_Examples}(e) depicts the degradation of the capacity regions due to various strengths of $s_{b,2}$. From the above discussions, we make the following observations.
\begin{itemize}
\item [1)] The strength of the interference that is known to the strongest transmitter, that is, $s_{b,1}$, does not affect the capacity region, as in the single-user point-to-point case.
\item [2)] Based on the interference cancellation capability of each transmitter (its transmit power), the bit levels of $y_{b}$ can be partitioned into $K$ layers (here $K=2$): layer 1, consisting of the third and the fourth LSB levels, and layer 2, consisting of the first and second levels. In the bottom layer $2$, both interferences caused by $s_{b,1}$ and $s_{b,2}$ can be completely cancelled. In this layer both users share $n_2$ bit levels. On the other hand, in the top layer $1$, only the interference caused by $s_{b,1}$ can be cancelled, while that caused by $s_{b,2}$ cannot. Hence in this layer user 1 can only use $\lp n_1-m_2\rp^+$ levels.
\end{itemize}

These observations lead to a natural way for establishing achievability, which is detailed in Section~\ref{subsec_LDCAchieve}. For the converse, the above discussion gives the intuitive explanation why the lack of knowledge about $s_{b,2}$ at transmitter 1 degrades the capacity region. In Section~\ref{subsec_LDCconverse} we give a formal converse proof.

\subsection{Achievability}\label{subsec_LDCAchieve}
Each transmitter, say $i$, cancels the interference it knows, $s_{b,i}$, as much as it can. If $m_i\le n_i$, then $s_{b,i}$ can be completely canceled. If $m_i > n_i$, then the top most $\lp m_i-n_i\rp$ levels of $s_{b,i}$ cannot be removed, and the bit levels of $y_b$ occupied by this chunk can never be used to convey data by any user. Since the channel is linear and the interferences are additive, the effect of interference cancelation remains for other users.

Superimposed upon interference cancellation, the scheme consists of $K$ layers. Layer $i$ is from the $\lp n_{i+1}+1\rp$-th level of LSB to the $n_i$-th level at the receiver, $i\in[1:K]$. In layer $i$, user $[1:i]$ can transmit. Therefore, we have the following achievable rates in layer $i$, $i\in[1:K]$: $r^{(i)}_{[1:i]} \ge 0$ satisfying
\begin{align}
\sum_{l=1}^{i} r^{(i)}_l \le \lp n_i - \max\lbp m_{[i+1:K]}, n_{i+1}\rbp\rp^+.
\end{align}

User $i$'s rate is the aggregate of its rates from layer $i$ to layer $K$: $r_i = \sum_{l=i}^K r^{(l)}_i$. Apply Fourier-Motzkin elimination we establish Lemma \ref{lem_LDCachieve}.

\subsection{Converse Proof}\label{subsec_LDCconverse}
Next we prove the outer bounds in Lemma \ref{lem_LDCconverse}. 
\begin{proof}
Let 
\begin{align}
y_{b,k} := \sum_{i=k}^K \underbrace{A^{q-n_i}x_{b,i}}_{x_i} + \sum_{i=k}^K \underbrace{A^{q-m_i}s_{b,i}}_{s_i}.
\end{align}
Here we use $x_i$ to denote $A^{q-n_i}x_{b,i}$ and $s_i$ to denote $A^{q-m_i}s_{b,i}$ for notational convenience. It is easy to distinguish these notations from those in the original Gaussian model based on the context.

If $r_{[1:K]}$ is achievable, for any $k \in [1:K]$ by Fano's inequality and data processing inequality, we have
\begin{align}
&N\lp \sum_{i=k}^K r_i -\epsilon_N\rp\\ 
&\le I\lp w_{[k:K]}; y_b^N | w_{[1: k-1]}\rp\\
&\overset{\aaaa}{\le} I\lp w_{[k:K]}; y_b^N | w_{[1: k-1]} , s^N_{[1:k-1]}\rp\\
&= H\lp y_b^N | w_{[1: k-1]} , s^N_{[1:k-1]}\rp - H\lp y_b^N | w_{[1: K]} , s^N_{[1:k-1]}\rp\\
&= H\lp y_{b,k}^N | w_{[1: k-1]} , s^N_{[1:k-1]}\rp - H\lp y_{b,k}^N | w_{[1: K]} , s^N_{[1:k-1]}\rp\\
&\overset{\bbbb}{=} I\lp w_{[k:K]}; y_{b,k}^N\rp = I\lp w_{[k:K]}, s_{[k:K]}^N; y_{b,k}^N\rp - I\lp s_{[k:K]}^N; y_{b,k}^N| w_{[k:K]}\rp\\
&\overset{\cccc}{=} H\lp y_{b,k}^N\rp - \sum_{i=k}^K H\lp s_{i}^N\rp + H\lp s_{[k:K]}^N | y_{b,k}^N, w_{[k:K]}\rp\\
&= H\lp y_{b,k}^N\rp - \sum_{i=k}^K H\lp s_i^N\rp + \sum_{i=k}^K H\lp s_i^N | y_{b,k}^N, w_{[k:K]}, s_{[k:i-1]}^N\rp\\
&\overset{\dddd}{\le} H\lp y_{b,k}^N\rp - H\lp s_k^N\rp + H\lp s_k^N|y_{b,k}^N\rp - \sum_{i=k+1}^{K} H\lp s_i^N\rp + \sum_{i=k+1}^{K} H\lp s_i^N | y_{b,i}^N, w_{[i:K]}\rp\\
&\overset{\eeee}{\le} H\lp y_{b,k}^N| s_k^N\rp - \sum_{i=k+1}^{K} H\lp s_i^N\rp + \sum_{i=k+1}^{K} \min\lbp H\lp s_i^N \rp, H\lp x_i^N + \sum_{l=i+1}^{K} \lp x_l^N+s_l^N\rp \rp \rbp\\
&\le N\lbp \max\lbp m_{[k+1:K]},n_k\rbp -  \sum_{i=k+1}^{K} m_i +  \sum_{i=k+1}^{K} \min\lbp m_i,\max\lbp m_{[i+1:K]}, n_i\rbp \rbp \rbp \label{eq_upterm},
% = N \ol{\msf{R}}_k\lp n_{[1:k]}, m_{[1:k-1]}\rp,
\end{align}
where $\epsilon_N\rightarrow 0$ as $N\rightarrow \infty$. (a) is due to the facts that conditioning reduces entropy and that $s^N_{[1:k-1]}$ is independent of $w_{[k:K]}$. (b) is due to the fact that $\lp w_{[k:K]}, s_{[k:K]}^N,y_{b,k}^N\rp$ and $\lp w_{[1:k-1]}, s_{[1:k-1]}^N\rp$ are independent. (c) is due to the fact that $\lbp w_{[k:K]},s_{[k:K]}^N \rbp$ are mutually independent and $y_{b,k}^N$ is a function of $\lbp w_{[k:K]},s_{[k:K]}^N \rbp$. (d) is due to conditioning reduces entropy and the fact that $\lp w_{[i:K]}, s_{[i:K]}^N,y_{b,i}^N\rp$ and $\lp w_{[k:i-1]}, s_{[k:i-1]}^N\rp$ are independent. (e) is due to the fact that $y_{b,i}^N = x_i^N + s_i^N + \sum_{l=i+1}^{K} \lp x_l^N+s_l^N\rp$.

It is straightforward to see that $\eqref{eq_upterm} = N \ol{\msf{r}}_k\lp n_{[k:K]}, m_{[k+1:K]};K\rp$. Proof complete.
\end{proof}

\subsection{Implication on the Gaussian Problem}
By investigating the binary expansion model, we gain intuitions about how to solve the original Gaussian problem. For the outer bounds, we will mimic the proof in Section \ref{subsec_LDCconverse}. For the achievability in the binary expansion model, interference cancellation is realized by simply subtracting interferences from the transmit signals. Due to linearity of the channel and the fact that there is no interaction among different bit levels, if an interference, say, a component of $s_{b,1}$, is cancelled by transmitter 1, it will remain cancelled for other users as well. To realize such distributed interference cancellation in the Gaussian scenario, however, Philosof \etal \cite{PhilosofZamir_07} show that Gelfand-Pinsker scheme based on Gaussian random binning is not sufficient. Instead, they propose a modulo-lattice scheme which can carry out this task. Motivated by the layered nature in the achievability of the binary expansion model, we propose a \emph{layered} modulo-lattice scheme, generalized from the single-layer scheme in \cite{PhilosofZamir_07}, to realize distributed interference cancellation in all layers, and show that it achieves the capacity region to within a constant number of bits.

\section{Layered Modulo-Lattice Scheme}\label{sec_Lattice}
%In this section, we first briefly review the single-layer scheme in \cite{PhilosofZamir_07} as well as preliminary on lattices. Then we propose our layered modulo-lattice scheme.
In this section we first give a brief review on lattices and propose the modulo-lattice scheme used in each layer of our layered architecture. Then we connect all layers, describe the overall architecture, and derive the achievable rates in all layers.

\subsection{A Primer on Lattices}
Before introducing the modulo-lattice scheme, first we give some basic definitions and facts about lattices. For more detailed introduction, please refer to \cite{PhilosofZamir_07} and the references therein. For completeness, the following basic and useful facts adapted from \cite{PhilosofZamir_07} are introduced.

%\begin{definition}[Lattices]
An $N$-dimensional lattice $\Lambda$ is defined as
\begin{align}
\Lambda := \lbp \mb{l}=B\mb{i}: \mb{i}\in\mathbb{Z}^N\rbp,
\end{align}
where $B\in \mathbb{R}^{N\times N}$ is non-singular.
%\end{definition}
By definition, the origin $\mb{0}\in\Lambda$.

A natural procedure associated to lattice $\Lambda$ is to quantize points in $\mathbb{R}^N$ to the nearest lattice point.
%\begin{definition}[Nearest neighbor quantizer]
The nearest neighbor quantizer associated with lattice $\Lambda$ is defined as
\begin{align}
Q_{\Lambda}\lp\mb{x}\rp := \arg\min_{\mb{l}\in\Lambda} \| \mb{x} - \mb{l}\|,\ \forall \mb{x}\in\mathbb{R}^N.
\end{align}
%\end{definition}
Here $\|\cdot\|$ denote the Euclidean norm.

Another natural procedure is to take the modulo on a lattice.
%\begin{definition}[Modulo operation]
For any $\mb{x}\in\mathbb{R}^N$, its modulo on lattice $\Lambda$ is the ``quantization error"
\begin{align}
\mb{x} \bmod \Lambda := \mb{x}-Q_{\Lambda}\lp\mb{x}\rp.
\end{align}
%\end{definition}
Note that the modulo-lattice operation satisfies the distributive property: for any $\mb{x},\mb{y}\in\mbb{R}^N$,
\begin{align}
\lb\lp\mb{x} \bmod \Lambda\rp + \mb{y}\rb \bmod \Lambda = \lb \mb{x} + \mb{y}\rb \bmod \Lambda.
\end{align}

%\begin{definition}[Basic Voronoi region]
The basic Voronoi region of lattice $\Lambda$ is defined as
\begin{align}
\mcal{V} := \lbp \mb{x}\in \mathbb{R}^N: Q_{\Lambda}\lp\mb{x}\rp = \mb{0}\rbp.
\end{align}
%\end{definition}
We denote the volume of $\mcal{V}$ by $V$, $V=\int_{\mcal{V}} d\mb{x}$.

%\begin{definition}[Determinant]
%We denote the volume of $\mcal{V}$ by $V$. Note that $V = \det G$ and hence we also call $V$ the \emph{determinant} of lattice $\Lambda$; that is,
%\begin{align}
%\det\lp \Lambda \rp := V = \det G.
%\end{align}
%\end{definition}

The second moment of the a lattice $\Lambda$ is defined by the second moment per dimension of a uniform distribution over the basic Voronoi region $\mcal{V}$:
\begin{align}
\sigma_\Lambda^2 := \frac{1}{N}\frac{\int_{\mcal{V}}\|\mb{x}\|^2d\mb{x}}{V}.
\end{align}
The normalized second moment is defined by 
\begin{align}
G\lp \Lambda\rp := \frac{\sigma_\Lambda^2}{V^{2/N}}.	
\end{align}
Note that the normalized second moment of a lattice is always lower bounded by $\frac{1}{2\pi e}$ \cite{ZamirFeder_96}.

%In \cite{ZamirFeder_96} it is known that there exists a sequence of lattices $\lbp \Lambda_N, N\in\mathbb{N}\rbp$, that is \emph{good for quantization}, in the sense that
%\begin{align}
%\lim_{N\ra\infty} G\lp \Lambda_N \rp = \frac{1}{2\pi e}.
%\end{align}

The following lemmas \cite{ZamirFeder_96} turn out to be useful for computing achievable rates.
\begin{lemma}\label{lem_Entropy}
For a given $N$-dimensional lattice $\Lambda$ with basic Voronoi region $\mcal{V}$, if random vector $\mb{Y} \sim \mathrm{Unif}\lp \mcal{V}\rp$, then
\begin{align}
h\lp \mb{Y}\rp = \log\lp V\rp = \frac{N}{2}\log\lp \frac{\sigma_\Lambda^2}{G\lp \Lambda\rp}\rp.
\end{align}
\end{lemma}

\begin{lemma}\label{lem_White}
Consider an $N$-dimensional lattice $\Lambda$ has the minimal normalized second moment. If random vector $\mb{Y} \sim \mathrm{Unif}\lp \mcal{V}\rp$, then its covariance matrix is white: $K_{\mb{Y}} = \sigma_{\Lambda}^2 I_{N}$. Moreover, there exists a sequence of such lattices $\lbp \Lambda_N, N\in\mathbb{N}\rbp$, that is \emph{good for quantization}, in the sense that they attain the lower bound $\frac{1}{2\pi e}$ as $N\ra\infty$:
\begin{align}
\lim_{N\ra\infty} G\lp \Lambda_N \rp = \frac{1}{2\pi e}. \label{eq_Limit}
\end{align}
\end{lemma}

%Combining this lemma with the existence of lattices that are good for quantization, we have the following corollary:
%\begin{corollary}
%\end{corollary}

\subsection{Modulo-Lattice Scheme in Each Layer}
In each layer, we shall use the following canonical modulo-lattice scheme, which is a simplified version of that in \cite{PhilosofZamir_07}.

Consider a generic layer $k$ where the subset of participating users is $S^{(k)}\subseteq[1:K]$. The received signal can be written as
\begin{align}
\mb{y} = \sum_{i\in S^{(k)}} \mb{x}^{(k)}_i + \sum_{i\in S^{(k)}} \mb{s}^{(k)}_i + \mb{z}^{(k)}, \label{eq_DecoderInput}
\end{align}
where $\mb{x}^{(k)}_i$ denotes user $i$'s transmit signal in this layer, $\mb{s}^{(k)}_i$ denotes the \emph{interference} in this layer that is known to user $i$, and $\mb{z}^{(k)}$ denotes the effective aggregate \emph{noise} in this layer. All the transmit signals, interferences, and the noise are mutually independent. The difference between interference and noise is that, interference is mitigated using side information precoding, while noise cannot and hence persists in the received signal. As we shall see in the overall architecture of our layered strategy, interferences $\mb{s}^{(k)}_i$ and effective noise $\mb{z}^{(k)}$ will contain the signals sent in other layers, and hence is not necessary Gaussian.

The canonical modulo-lattice scheme is configured by three parameters: (1) an $N$-dimensional lattice $\Lambda^{(k)}$, (2) its second moment $\Theta^{(k)}$, and (3) $S^{(k)}\subseteq[1:K]$, the subset of users participating in the transmission. For each user $i\in S^{(k)}$, its corresponding sub-encoder in this layer uses lattice $\Lambda^{(k)}$ with second moment $\Theta^{(k)}$ and basic Voronoi region $\mcal{V}^{(k)}$ to modulate its sub-message $w_i^{(k)}$ in this layer. Its codeword, $\mb{v}_i^{(k)}$, is generated according to $\mathrm{Unif}\lp\mcal{V}^{(k)}\rp$ with rate $R^{(k)}_{i}$. The transmit signal $\mb{x}_i^{(k)}$ is generated according to the following modulo-lattice operation:
\begin{align}
\mb{x}_i^{(k)} = \lb \mb{v}_i^{(k)} - \alpha^{(k)}\mb{s}_i^{(k)} - \mb{d}_i^{(k)} \rb \bmod \Lambda^{(k)}, \label{eq_ModTx}
\end{align}
where $\mb{d}_i^{(k)} \sim \mathrm{Unif}\lp\mcal{V}^{(k)}\rp$ independent of everything else, is the dither known at the receiver (common randomness).

%\begin{figure}[htbp]
%{\center
%\subfigure[]{\includegraphics[width=3in]{ModENC.pdf}}
%\subfigure[]{\includegraphics[width=3in]{ModDEC.pdf}}
%%\includegraphics[width=2in]{ENC_k.pdf}
%\caption{Modulo-lattice transceiver in layer $k$.}
%\label{fig_TxRxk}
%}
%\end{figure}

The corresponding decoder in this layer, upon receiving $\mb{y}$, first multiplies $\mb{y}$ by $\alpha^{(k)}$, adds the dithers back, and then takes the modulo $\Lambda^{(k)}$ operation. The output becomes
\begin{align}
\mb{y}^{(k)} &= \lb \alpha^{(k)}\mb{y} + \sum_{i\in S^{(k)}} \mb{d}_i^{(k)}\rb \bmod \Lambda^{(k)} \label{eq_Decode1}\\
&= \lb \mb{y} - \lp1-\alpha^{(k)}\rp\mb{y} + \sum_{i\in S^{(k)}} \mb{d}_i^{(k)}\rb \bmod \Lambda^{(k)}\\
&= \lb\begin{array}{l} \sum_{i\in S^{(k)}} \lb \mb{v}_i^{(k)} - \alpha^{(k)}\mb{s}_i^{(k)} - \mb{d}_i^{(k)} \rb \bmod \Lambda^{(k)} + \sum_{i\in S^{(k)}} \mb{s}^{(k)}_i + \mb{z}^{(k)}\\
- \lp1-\alpha^{(k)}\rp \lp \sum_{i\in S^{(k)}} \mb{x}^{(k)}_i + \sum_{i\in S^{(k)}} \mb{s}^{(k)}_i + \mb{z}^{(k)}\rp + \sum_{i\in S^{(k)}} \mb{d}_i^{(k)}\end{array}\rb \bmod \Lambda^{(k)}\\
&= \lb \sum_{i\in S^{(k)}} \mb{v}_i^{(k)} + \alpha^{(k)}\mb{z}^{(k)} - \lp1-\alpha^{(k)}\rp \sum_{i\in S^{(k)}} \mb{x}^{(k)}_i \rb \bmod \Lambda^{(k)}\\
&= \lb \sum_{i\in S^{(k)}} \mb{v}_i^{(k)} + \mb{z}_{\rm{eq}}^{(k)} \rb \bmod \Lambda^{(k)}, \label{eq_Decode2}
\end{align}
where $\mb{z}_{\rm{eq}}^{(k)} := \alpha^{(k)}\mb{z}^{(k)} - \lp1-\alpha^{(k)}\rp \sum_{i\in S^{(k)}} \mb{x}^{(k)}_i$ denotes the \emph{effective noise} in the \emph{effective modulo-lattice channel} in layer $k$. From the first line, due to dithers the output signal $\mb{y}^{(k)} \sim \mathrm{Unif}\lp \mcal{V}^{(k)}\rp$. Moreover, $\mb{z}_{\rm{eq}}^{(k)}$ is independent of $\mb{v}_{S^{(k)}}$ due to dithering \cite{PhilosofZamir_07}.

%With the above transceiver architecture, we transform the original Gaussian multiple access channel into a modulo-lattice $\Lambda^{(k)}$

%The achievable rates of the scheme in this layer can be derived following the same line of analysis as in \cite{PhilosofZamir_07}: non-negative rate tuples $R^{(k)}_{S^{(k)}}$ is achievable, if
%\begin{align}
%N\sum_{i\in S^{(k)}} R_i^{(k)} &\le I\lp \mb{v}_{S^{(k)}}; \mb{y}^{(k)}\rp = h\lp \mb{y}^{(k)}\rp - h\lp \lb \sum_{i\in S^{(k)}} \mb{v}_i^{(k)} + \mb{z}_{\rm{eq}}^{(k)} \rb \bmod \Lambda^{(k)}\Bigg| \mb{v}_{S^{(k)}}\rp.
%\end{align}

%Since $\mb{y}^{(k)} \sim \mathrm{Unif}\lp \mcal{V}^{(k)}\rp$, due to Lemma \ref{lem_Entropy} $h\lp \mb{y}^{(k)}\rp = \frac{N}{2}\log\lp \frac{\sigma_{\Lambda^{(k)}}^2}{G\lp \Lambda^{(k}\rp}\rp$. Moreover, since the modulo operation only reduces the entropy,
%\begin{align}
%h\lp \lb \sum_{i\in S^{(k)}} \mb{v}_i^{(k)} + \mb{z}_{\rm{eq}}^{(k)} \rb \bmod \Lambda^{(k)}\Bigg| \mb{v}_{S^{(k)}}\rp &\le h\lp \sum_{i\in S^{(k)}} \mb{v}_i^{(k)} + \mb{z}_{\rm{eq}}^{(k)} \Bigg| \mb{v}_{S^{(k)}}\rp\\
%&= h\lp \mb{z}_{\rm{eq}}^{(k)} \rp\\
%&\overset{\aaaa}{\le} \frac{N}{2}\log\lp 2\pi e\lp \sigma^2_{z^{(k)}} + |S^{(k)}|\sigma_{\Lambda^{(k)}}^2 \rp \rp.
%\end{align}
%(a) is due to the fact that Gaussian dis

\subsection{Overall Architecture}
Now we are ready to describe the overall architecture of our layered modulo-lattice scheme. 

{\flushleft \it Encoding}\par
For encoding, we shall use an inductive way to describe from the top layer $1$ to the bottom layer $K$, which corresponds to the order of encoding.

{\flushleft 1) Layer $1$:}
In this layer, the set of participating users is $S^{(1)} = \{1\}$. We choose the modulation lattice $\Lambda^{(1)}$ to be the one that attaining the minimal normalized second moment with second moment $\Theta^{(1)} = P_1-P_2$. The interference $\mb{s}_1^{(1)} = \mb{s}_1$. The sub-encoder $\mcal{E}_{1}^{(1)}$ generates $\mb{x}_{1}^{(1)}$ based on \eqref{eq_ModTx} for $k=1$, and feeds $\mb{s}_{1}^{(2)} := \mb{s}_{1}^{(1)} + \mb{x}_{1}^{(1)}$ to the next-layer sub-encoder $\mcal{E}_{1}^{(2)}$.
%The noise 
%\begin{align}
%\mb{z}^{(1)} = \mb{z} + \sum_{l=2}^{K}\sum_{i\in S^{(l)}} \mb{x}_{i}^{(l)}.
%\end{align}

{\flushleft 2) Layer $k,1<k<K$:}
The set of participating users is $S^{(k)} = [1:k]$. The modulation lattice $\Lambda^{(k)}$ is the one that attaining the minimal normalized second moment with second moment $\Theta^{(k)} = P_k-P_{k+1}$. The known interference $\mb{s}_i^{(k)} = \mb{s}_i^{(k-1)} + \mb{x}_{i}^{(k-1)}$, for all $i\in [1:k-1]$, and $\mb{s}_k^{(k)} = \mb{s}_k$. The sub-encoder $\mcal{E}_{i}^{(k)}$ generates $\mb{x}_{i}^{(k)}$ based on \eqref{eq_ModTx}, and feeds $\mb{s}_{i}^{(k+1)} := \mb{s}_{i}^{(k)} + \mb{x}_{i}^{(k)}$ to the next-layer sub-encoder $\mcal{E}_{i}^{(k+1)}$, for all $i\in S^{(k)}=[1:k]$.
%The noise 
%\begin{align}
%\mb{z}^{(k)} = \mb{z} + \sum_{l=k+1}^{K}\sum_{i\in S^{(l)}} \mb{x}_{i}^{(l)}.
%\end{align}
%where $\lbp\mb{x}_{i}^{(l)}\rbp$ are generated in lower layers $[k+1:K]$.

{\flushleft 3) Layer $K$:}
The set of participating users is $S^{(K)} = [1:K]$. The modulation lattice $\Lambda^{(K)}$ is the one that attaining the minimal normalized second moment with second moment $\Theta^{(K)} = P_K$. The known interference $\mb{s}_i^{(K)} = \mb{s}_i^{(K-1)} + \mb{x}_{i}^{(K-1)}$, for all $i\in [1:K-1]$, and $\mb{s}_K^{(K)} = \mb{s}_K$. The sub-encoder $\mcal{E}_{i}^{(K)}$ generates $\mb{x}_{i}^{(K)}$ based on \eqref{eq_ModTx} for $k=K$.

{\flushleft \it Decoding}\par
The receiver decodes layer $k\in[1:K]$ with sub-decoder $\mcal{D}^{(k)}$. Unlike the sequential operation at the sub-encoders, these sub-decoders work \emph{in parallel}. $\mcal{D}^{(k)}$ takes the received signal $\mb{y}$ as input, which can be written as \eqref{eq_DecoderInput}, and takes the operation in $\eqref{eq_Decode1}-\eqref{eq_Decode2}$ to generate $\mb{y}^{(k)}$. With the above-mentioned encoding operations, the effective noise
\begin{align}
\mb{z}^{(k)} &= \lbp\begin{array}{ll}\mb{z} + \sum_{l=k+1}^{K}\lp \mb{s}_l + \sum_{i=1}^l \mb{x}_{i}^{(l)}\rp, & 1\le k \le K-1\\ \mb{z}, & k=K
\end{array}\right.\\
\Cov\lb \mb{z}^{(k)} \rb &= \lbp\begin{array}{ll} \lp N_o + \sum_{l=k+1}^K \lp Q_l + l\Theta^{(l)}\rp\rp I_N, & 1\le k \le K-1\\ N_oI_N, & k=K
\end{array}\right. := N^{(k)}I_N, \label{eq_Cov}
\end{align}
due to our choice of lattices and Lemma \ref{lem_White}. $N^{(k)}$ denotes the effective per-symbol noise variance in layer $k$. Due to dithering, indeed $\mb{x}^{(k)}_{[1:k]}$, $\mb{s}^{(k)}_{[1:k]}$, and $\mb{z}^{(k)}$ are mutually independent. Based on $\mb{y}^{(k)}$, it performs joint typicality decoding as in standard MAC to find $\mb{v}^{(k)}_{S^{(k)}}$, where $S^{(k)} = [1:k]$.

The overall architecture of transmitters and receiver is depicted in Fig.~\ref{fig_ENCk}.

\begin{figure}[htbp]
{\center
%\subfigure[]{\includegraphics[width=3in]{ModENC.pdf}}
%\subfigure[]{\includegraphics[width=3in]{ModENC.pdf}}\\
%\subfigure[]{\includegraphics[width=3in]{ModDEC.pdf}}
%\subfigure[]{\includegraphics[width=2in]{ENC_k.pdf}}
%\subfigure[]{\includegraphics[width=5in]{Arch.pdf}}
%\subfigure[]{\includegraphics[width=6in]{Mod.pdf}}
\includegraphics[width=6in]{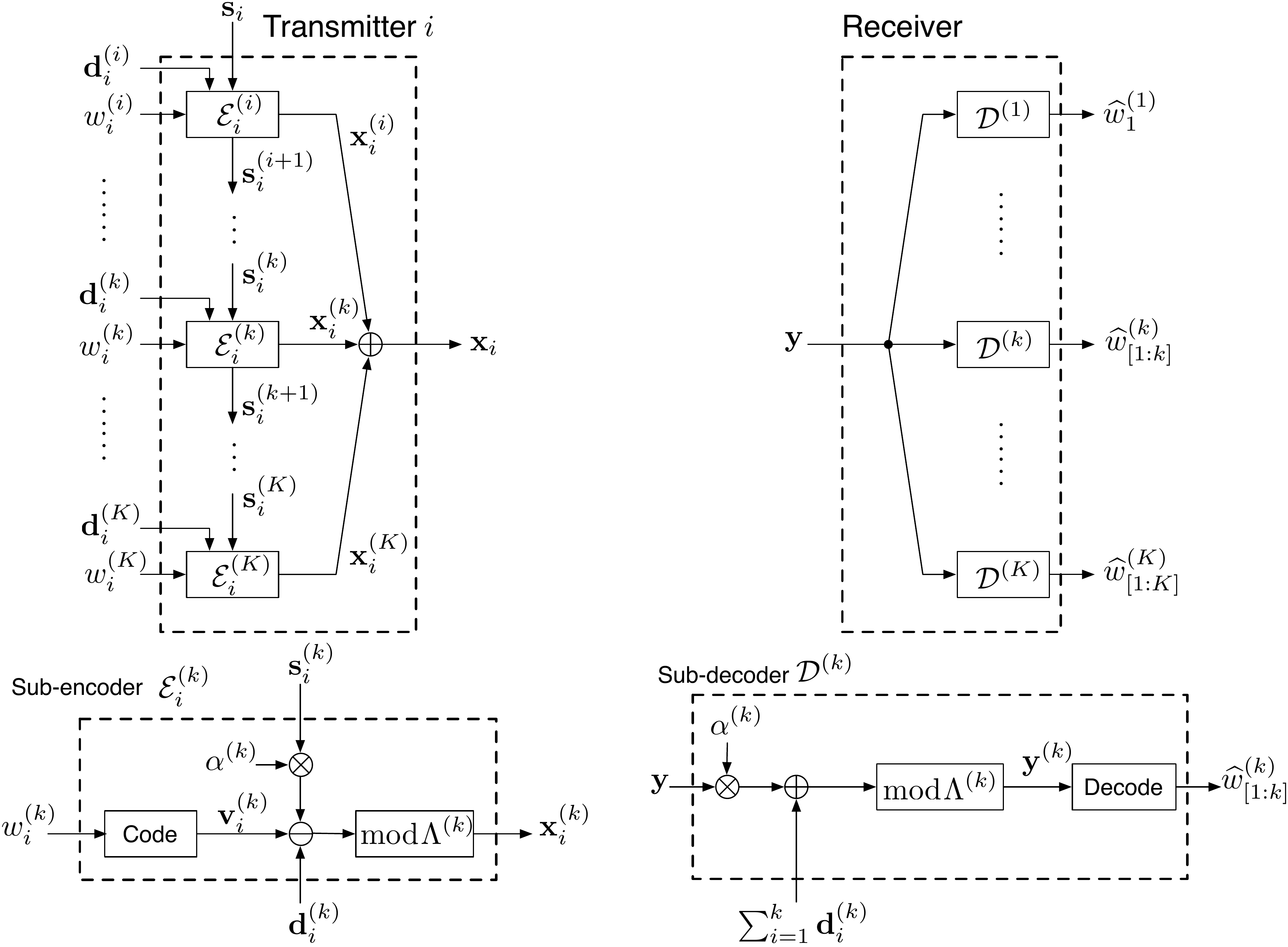}
\caption{Transmitter and Receiver Architecture}
\label{fig_ENCk}
}
\end{figure}

\subsection{Achievable Rates in Each Layer}\label{subsec_LatticeRate}

The achievable rates of the scheme in this layer can be derived following the same line of analysis as in \cite{PhilosofZamir_07}: non-negative rate tuples $R^{(k)}_{[1:k]}$ is achievable, if
\begin{align}
N\sum_{i=1}^k R_i^{(k)} &\le I\lp \mb{v}^{(k)}_{[1:k]}; \mb{y}^{(k)}\rp = h\lp \mb{y}^{(k)}\rp - h\lp \lb \sum_{i=1}^k \mb{v}_i^{(k)} + \mb{z}_{\rm{eq}}^{(k)} \rb \bmod \Lambda^{(k)}\Bigg| \mb{v}^{(k)}_{[1:k]}\rp. \label{eq_Achieve}
\end{align}

Since $\mb{y}^{(k)} \sim \mathrm{Unif}\lp \mcal{V}^{(k)}\rp$, due to Lemma \ref{lem_Entropy}, the first term $h\lp \mb{y}^{(k)}\rp = \frac{N}{2}\log\lp \frac{\Theta^{(k)}}{G\lp \Lambda^{(k)}\rp}\rp$. Moreover, since the modulo operation only reduces the entropy, the second term can be upper bounded as follows:
\begin{align}
&h\lp \lb \sum_{i=1}^k \mb{v}_i^{(k)} + \mb{z}_{\rm{eq}}^{(k)} \rb \bmod \Lambda^{(k)}\Bigg| \mb{v}^{(k)}_{[1:k]}\rp\\
&\le h\lp \sum_{i=1}^k \mb{v}_i^{(k)} + \mb{z}_{\rm{eq}}^{(k)} \Bigg| \mb{v}^{(k)}_{[1:k]}\rp\\
&\overset{\aaaa}{=} h\lp \mb{z}_{\rm{eq}}^{(k)} \rp\\
&\overset{\bbbb}{\le} \frac{N}{2}\log\lp 2\pi e\lp \lp\alpha^{(k)}\rp^2N^{(k)} + \lp1-\alpha^{(k)}\rp^2k\Theta^{(k)} \rp \rp.
\end{align}
(a) is due to the fact that $\mb{z}_{\rm{eq}}^{(k)}$ and $\mb{v}^{(k)}_{[1:k]}$ are independent. (b) is due to the fact that Gaussian distribution is the entropy maximizer for a given covariance matrix, and that the covariance matrix of $\mb{z}_{\rm{eq}}^{(k)} = \alpha^{(k)}\mb{z}^{(k)} - \lp1-\alpha^{(k)}\rp \sum_{i=1}^k \mb{x}^{(k)}_i$ is
\begin{align}
\Cov\lb \mb{z}_{\rm{eq}}^{(k)} \rb = \lp\alpha^{(k)}\rp^2N^{(k)} I_N +  \lp1-\alpha^{(k)}\rp^2 k\Theta^{(k)} I_N, 
\end{align}
based on \eqref{eq_Cov} and Lemma~\ref{lem_White}.

Hence, combining the above two, we obtain a lower bound on the right-hand side of \eqref{eq_Achieve}:
\begin{align}
&\frac{N}{2}\log\lp \frac{\Theta^{(k)}}{G\lp \Lambda^{(k)}\rp}\rp - \frac{N}{2}\log\lp 2\pi e\lp \lp\alpha^{(k)}\rp^2N^{(k)} + \lp1-\alpha^{(k)}\rp^2k\Theta^{(k)} \rp \rp\\
=\ &N\lbp \frac{1}{2}\log\lp \frac{\Theta^{(k)}}{\lp\alpha^{(k)}\rp^2N^{(k)} + \lp1-\alpha^{(k)}\rp^2k\Theta^{(k)}}\rp - \frac{1}{2}\log\lp 2\pi e G\lp \Lambda^{(k)}\rp \rp \rbp.
\end{align}

Based on Lemma \ref{lem_White}, there exists a sequence of lattices satisfying \eqref{eq_Limit}, and therefore all non-negative rates satisfying
\begin{align}
\sum_{i=1}^k R_i^{(k)} &\le \frac{1}{2}\log^+\lp \frac{\Theta^{(k)}}{\lp\alpha^{(k)}\rp^2N^{(k)} + \lp1-\alpha^{(k)}\rp^2k\Theta^{(k)}}\rp
\end{align}
are achievable in layer $k$, $k\in[1:K]$. Note that the optimal choice of $\alpha^{(k)}$ is the MMSE coefficient $\alpha^{(k)} = \frac{\lp N^{(k)}\rp\lp k\Theta^{(k)}\rp}{N^{(k)} + k\Theta^{(k)}}$, and the resulting rate constraint is
\begin{align}
\sum_{i=1}^k R_i^{(k)} &\le \frac{1}{2}\log^+\lp \frac{1}{k} + \frac{\Theta^{(k)}}{N^{(k)}}\rp\\
&= \frac{1}{2}\log^+\lp \frac{1}{k} + \frac{P_k-P_{k+1}}{N_o + \sum_{l=k+1}^{K} Q_l + \sum_{l=k+1}^{K-1} l\lp P_l - P_{l+1}\rp + KP_K}\rp\\
&= \frac{1}{2}\log^+\lp \frac{1}{k} + \frac{P_k-P_{k+1}}{N_o + (k+1)P_{k+1} + Q_{k+1} + \sum_{j=k+2}^K \lp P_j+Q_j\rp}\rp\\
&= \frac{1}{2}\log^+\lp \frac{N_o + kP_{k} + \sum_{j=k+1}^K \lp P_j+Q_j\rp}{k\lp N_o + kP_{k+1} + \sum_{j=k+1}^K \lp P_j+Q_j\rp \rp}\rp.
%&= \frac{1}{2}\log^+\lp \frac{1 + k\SNR_{k} + \sum_{j=k+1}^K \lp \SNR_j+\INR_j\rp}{k\lp 1 + k\SNR_{k+1} + \sum_{j=k+1}^K \lp \SNR_j+\INR_j\rp \rp}\rp.
\end{align}
For notational convenience, we denote $P_{K+1} = \SNR_{K+1}=0$.

In the next section, we derive outer bounds based on similar proof techniques as in the binary expansion model (Section~\ref{subsec_LDCconverse}), derive inner bounds based on the discussion above, and show that they are within a constant number of bits to one another.

\section{Constant Gap to Capacity}\label{sec_ConstGap}
The main result is summarized in the following lemmas and theorem.

\begin{lemma}[Outer Bounds]\label{lem_GCconverse} 
If $R_{[1:K]}\ge 0$ is achievable, it satisfies the following: for all $k \in [1:K]$,
\begin{align}
%\sum_{i=1}^k R_i \le \ol{\msf{R}}_k\lp n_{[1:k]}, m_{[1:k-1]}\rp,
\sum_{i=k}^K R_i \le \ol{\msf{R}}_k\lp \SNR_{[k:K]}, \INR_{[k+1:K]};K\rp,
\end{align}
where 
\begin{align}
%\ol{\msf{R}}_k\lp n_{[1:k]}, m_{[1:k-1]}\rp := \max\lbp m_{[1:k-1]},n_k\rbp - \sum_{i=1}^{k-1}\lp m_i - \max\lbp m_{[1:i-1]}, n_i\rbp\rp^+.
&\ol{\msf{R}}_k\lp \SNR_{[k:K]}, \INR_{[k+1:K]};K\rp := \lbp\begin{array}{l}\frac{1}{2}\log\lp 1+\sum_{i=k+1}^{K}2\lp \SNR_i+\INR_i\rp+\SNR_k\rp\\ - \sum_{i=k+1}^{K}\frac{1}{2}\log^+\lp\frac{\INR_i}{1+\sum_{l=i+1}^{K}2\lp\SNR_l+\INR_i\rp+\SNR_i}\rp\end{array}\rbp
%&\max\lbp m_{[k+1:K]},n_k\rbp - \sum_{i=k+1}^{K}\lp m_i - \max\lbp m_{[i+1:K]}, n_i\rbp\rp^+.
\end{align}
\end{lemma}
\begin{proof}
The technique is similar to the converse proof for the binary expansion model. See Appendix~\ref{app_Pf_lem_GCconverse} for detail.
\end{proof}

\begin{lemma}[Inner Bounds]\label{lem_GCachieve}
If $R_{[1:K]}\ge 0$ satisfies the following: for all $k \in [1:K]$,
\begin{align}
%\sum_{i=1}^k R_i \le \ul{\msf{R}}_k\lp n_{[1:k]}, m_{[1:k-1]}\rp,
\sum_{i=k}^K R_i \le \ul{\msf{R}}_k\lp \SNR_{[k:K]}, \INR_{[k+1:K]};K\rp
\end{align}
it is achievable. Here
\begin{align}
%\ul{\msf{R}}_k\lp n_{[1:k]}, m_{[1:k-1]}\rp := \sum_{i=1}^k \lp n_i - \max\lbp m_{[1:i-1]}, n_{i-1}\rbp\rp^+.
&\ul{\msf{R}}_k\lp \SNR_{[k:K]}, \INR_{[k+1:K]};K\rp := \begin{array}{l}\sum_{i=k}^K \frac{1}{2}\log^+\lp \frac{1 + i\SNR_{i} + \sum_{j=i+1}^K \lp \SNR_j+\INR_j\rp}{i\lp 1 + i\SNR_{i+1} + \sum_{j=i+1}^K \lp \SNR_j+\INR_j\rp \rp}\rp\end{array}
\end{align}
\end{lemma}
\begin{proof}
Based on Section~\ref{subsec_LatticeRate}, user $i$'s aggregate rate $R_i$ is the sum of rates in all layers in which it participates, that is, layer $i$ to layer $K$: $R_i = \sum_{l=i}^K R_i^{(l)}$. Applying Fourier-Motzkin elimination, we complete the proof.
\end{proof}

\begin{theorem}[Constant Gap to Capacity]\label{thm_GC}
$ $\\
The above inner and outer bounds are within $(K-k+1)\lp\log K+\frac{1}{2}\rp$ bits for user $k$, for all $k\in[1:K]$.
\end{theorem}
\begin{proof}
See Appendix~\ref{app_Pf_thm_GC}.
\end{proof}

%\subsection*{Approximate Capacity Using the Binary Expansion Model}
\begin{remark}
An alternative way to show the inner and outer bounds are within a constant is using the binary expansion model as an interface. Under the conversion in Definition~\ref{def_LDC}, it turns out that the outer bounds in Lemma~\ref{lem_LDCconverse} and Lemma~\ref{lem_GCconverse} are within a constant number of bits, as well as the inner bounds in Lemma~\ref{lem_LDCachieve} and Lemma~\ref{lem_GCachieve}. Then by Theorem~\ref{thm_LDC}, which shows that the inner and outer bounds match in the binary expansion model, it is immediate to establish the constant-gap-to-optimality result in the Gaussian scenario. Moreover, it justifies the usage of the binary expansion model in solving this problem, in the sense that its capacity region uniformly approximate that of the original Gaussian model. 
%The result is formally stated in the following theorem.
\end{remark}

\section{Conclusion}\label{sec_Conclude}
%We conclude the paper by first summarizing the main contribution and then describ some future extensions.
%\subsection{Summary}
Costa's landmark paper \cite{Costa_83} demonstrates that with proper precoding, in the point-to-point AWGN channel the effect of additive interference can be mitigated as if there were no interference, as long as the interference is known to the transmitter non-causally. In the multi-user scenario, however, when the interference is known partially to each node in the network, such conclusion no longer holds. Moreover, in the two-user doubly-dirty MAC, Philosof \etal \cite{PhilosofZamir_07} shows that a natural extension of Costa's Gaussian random binning scheme performs unboundedly worse than a lattice-based strategy. 

In this paper, we make a step further from \cite{PhilosofZamir_07}. We study the $K$-user Gaussian MAC with $K$ independent additive Gaussian interferences each of which known to exactly one transmitter non-causally, which is an extension of the two-user doubly-dirty MAC. With the help of a binary expansion model of the original problem, we propose a layered modulo-lattice scheme that realizes distributed interference cancellation, and characterize the capacity region to within a constant gap, for arbitrary channel parameters. 
%From the constant-gap-to-optimality result, except for the strongest transmitter, whenever the power of an interference is greater to the transmit power of transmitter that knows the interference, the damage of the interference cannot be completely removed. 
The binary expansion model uncovers the underlying layered structure of the original Gaussian problem, which leads naturally to the layered architecture and the converse proof.

%\subsection{Extensions}

%\section*{Acknowledgment}
%The author thanks Prof. David Tse for motivating this work and Prof. Mich\`{e}le Wigger for inspiring discussions.

\bibliographystyle{ieeetr}
%\bibliography{Ref}

\newpage

\appendices
\section{Proof of Lemma~\ref{lem_GCconverse}}\label{app_Pf_lem_GCconverse}
%\begin{proof}
Let 
\begin{align}
y_k := \sum_{i=k}^K x_i + \sum_{i=k}^K s_i + z.
\end{align}

If $R_{[1:K]}$ is achievable, for any $k \in [1:K]$ by Fano's inequality and data processing inequality, we have
\begin{align}
&N\lp \sum_{i=k}^K R_i -\epsilon_N\rp\\ 
&\le I\lp w_{[k:K]}; y^N | w_{[1: k-1]}\rp\\
&\overset{\aaaa}{\le} I\lp w_{[k:K]}; y^N | w_{[1: k-1]} , s^N_{[1:k-1]}\rp\\
&= h\lp y^N | w_{[1: k-1]} , s^N_{[1:k-1]}\rp - h\lp y^N | w_{[1: K]} , s^N_{[1:k-1]}\rp\\
&= h\lp y_k^N | w_{[1: k-1]} , s^N_{[1:k-1]}\rp - h\lp y_k^N | w_{[1: K]} , s^N_{[1:k-1]}\rp\\
&\overset{\bbbb}{=} I\lp w_{[k:K]}; y_k^N\rp = I\lp w_{[k:K]}, s_{[k:K]}^N; y_k^N\rp - I\lp s_{[k:K]}^N; y_k^N| w_{[k:K]}\rp\\
&\overset{\cccc}{=} h\lp y_k^N\rp - h\lp z^N\rp - \sum_{i=k}^K h\lp s_i^N\rp + h\lp s_{[k:K]}^N | y_k^N, w_{[k:K]}\rp\\
&= h\lp y_k^N\rp - h\lp z^N\rp - \sum_{i=k}^K h\lp s_i^N\rp + \sum_{i=k}^K h\lp s_i^N | y_k^N, w_{[k:K]}, s_{[k:i-1]}^N\rp\\
&\overset{\dddd}{\le} - h\lp z^N\rp + h\lp y_k^N\rp - h\lp s_k^N\rp + h\lp s_k^N|y_k^N\rp - \sum_{i=k+1}^{K} h\lp s_i^N\rp + \sum_{i=k+1}^{K} h\lp s_i^N | y_i^N, w_{[i:K]}\rp\\
&\overset{\eeee}{\le} - h\lp z^N\rp + h\lp y_k^N| s_k^N\rp - \sum_{i=k+1}^{K} h\lp s_i^N\rp\\
&\quad + \sum_{i=k+1}^{K} \min\lbp h\lp s_i^N \rp, h\lp x_i^N + \sum_{l=i+1}^{K} \lp x_l^N+s_l^N\rp + z^N \rp \rbp\\
&\overset{\ffff}{\le} N \ol{\msf{R}}_k\lp \SNR_{[k:K]}, \INR_{[k+1:K]};K\rp,
\end{align}
where $\epsilon_N\rightarrow 0$ as $N\rightarrow \infty$. (a) is due to the facts that conditioning reduces entropy and that $s^N_{[1:k-1]}$ is independent of $w_{[k:K]}$. (b) is due to the fact that $\lp w_{[k:K]}, s_{[k:K]}^N,y_k^N\rp$ and $\lp w_{[1:k-1]}, s_{[1:k-1]}^N\rp$ are independent. (c) is due to the fact that $\lbp w_{[k:K]},s_{[k:K]}^N \rbp$ are mutually independent and $y_k^N$ is a function of $\lbp w_{[k:K]},s_{[k:K]}^N \rbp$. (d) is due to conditioning reduces entropy and the fact that $\lp w_{[i:K]}, s_{[i:K]}^N,y_i^N\rp$ and $\lp w_{[k:i-1]}, s_{[k:i-1]}^N\rp$ are independent. (e) is due to the fact that $y_i^N = x_i^N + s_i^N + \sum_{l=i+1}^{K} \lp x_l^N+s_l^N\rp + z^N$. Finally, (f) is due to the fact that 
\begin{align*}
h\lp y_k^N| s_k^N\rp = h\lp x_k^N + \sum_{i=k+1}^K \lp x_i^N + s_i^N\rp + z^N\Big| s_k^N\rp \le h\lp x_k^N + \sum_{i=k+1}^K \lp x_i^N + s_i^N\rp + z^N\rp,
\end{align*}
Gaussian distribution maximizes the unconditional entropy, and $\Var\lb x_i^N + s_i^N\rb \le 2\Var\lb x_i^N\rb + 2\Var\lb s_i^N\rb$ for any $i$. Proof complete.
%\end{proof}

\section{Proof of Theorem~\ref{thm_GC}}\label{app_Pf_thm_GC}
We shall evaluate and upper bound the gap 
\begin{align}
\delta_k := \ol{\msf{R}}_k\lp \SNR_{[k:K]}, \INR_{[k+1:K]};K\rp - \ul{\msf{R}}_k\lp \SNR_{[k:K]}, \INR_{[k+1:K]}; K\rp.
\end{align}

For notational convenience, we denote $\Upsilon_i := \sum_{j=i}^K \lp\SNR_j+\INR_j\rp$.

First note that $\ul{\msf{R}}_k\lp \SNR_{[k:K]}, \INR_{[k+1:K]}; K\rp$ can be lower bounded by
\begin{align}
&\ul{\msf{R}}_k\lp \SNR_{[k:K]}, \INR_{[k+1:K]}; K\rp\\
&\ge \frac{1}{2}\log\lp 1 + k\SNR_k + \Upsilon_{k+1} \rp - \sum_{i=k}^K\frac{1}{2}\log i - \sum_{i=k+1}^{K}\frac{1}{2}\log\lp\frac{1 + i\SNR_i + \INR_i + \Upsilon_{i+1}}{1 + i\SNR_i + \Upsilon_{i+1}} \rp\\
&\ge \frac{1}{2}\log\lp 1 + \SNR_k + \Upsilon_{k+1} \rp - \sum_{i=k}^K\frac{1}{2}\log i - \sum_{i=k+1}^{K}\frac{1}{2}\log\lp\frac{1 + i\SNR_i + \INR_i + \Upsilon_{i+1}}{1 + i\SNR_i + \Upsilon_{i+1}} \rp.
\end{align}
Also,
\begin{align}
&\ol{\msf{R}}_k\lp \SNR_{[k:K]}, \INR_{[k+1:K]};K\rp\\
&= \frac{1}{2}\log\lp 1+2\Upsilon_{k+1}+\SNR_k\rp - \sum_{i=k+1}^{K}\frac{1}{2}\log^+\lp\frac{\INR_i}{1+2\Upsilon_{i+1}+\SNR_i}\rp\\
&\le \frac{1}{2}\log\lp 1+\Upsilon_{k+1}+\SNR_k\rp - \sum_{i=k+1}^{K}\frac{1}{2}\log^+\lp\frac{\INR_i}{1+\Upsilon_{i+1}+\SNR_i}\rp +\frac{1}{2}(K-k+1).
\end{align}
Hence,
\begin{align}
\delta_k 
&\le \sum_{i=k+1}^K\frac{1}{2}\log i
+\sum_{i=k+1}^{K}\frac{1}{2}\log\lp\frac{1 + \Upsilon_{i+1} + i\SNR_i + \INR_i}{1 + \Upsilon_{i+1} + i\SNR_i}\rp\\
&\quad - \sum_{i=k+1}^{K}\frac{1}{2}\log^+\lp\frac{\INR_i}{1+\Upsilon_{i+1}+\SNR_i}\rp +\frac{1}{2}(K-k+1)\\
&= \sum_{i=k}^K\frac{1}{2}\log i + \sum_{i=k+1}^{K} \lp \zeta_i - \xi_i\rp +\frac{1}{2}(K-k+1),
\end{align}
where $\zeta_i:=\frac{1}{2}\log\lp\frac{1 + \Upsilon_{i+1} + i\SNR_i + \INR_i}{1 + \Upsilon_{i+1} + i\SNR_i}\rp$ and $\xi_i:=\frac{1}{2}\log^+\lp\frac{\INR_i}{1+\Upsilon_{i+1}+\SNR_i}\rp$.

\begin{itemize}
%{\flushleft 1) If $\INR_i \le 1+\Upsilon_{i+1}+\SNR_i$}, then $\xi_i=0$, and $\zeta_i \le \frac{1}{2}\log\lp 1+1\rp = \frac{1}{2}$.
\item [1)]
If $\INR_i \le 1+\Upsilon_{i+1}+\SNR_i$, then $\xi_i=0$, and $\zeta_i \le \frac{1}{2}\log\lp 1+1\rp = \frac{1}{2}$.

%{\flushleft 2) If $\INR_i > 1+\Upsilon_{i+1}+\SNR_i$}, then
\item [2)]
If $\INR_i > 1+\Upsilon_{i+1}+\SNR_i$, then
\begin{align}
\zeta_i - \xi_i &= \frac{1}{2}\log\lp \frac{\lp1 + \Upsilon_{i+1} + i\SNR_i + \INR_i\rp\lp1+\Upsilon_{i+1}+\SNR_i\rp}{\INR_i\lp1 + \Upsilon_{i+1} + i\SNR_i\rp} \rp\\
&\le \frac{1}{2}\log\lp \frac{1 + \Upsilon_{i+1} + i\SNR_i + \INR_i}{\INR_i} \rp\\
&\le \frac{1}{2}\log\lp i+1\rp.
\end{align}
\end{itemize}

Therefore, combining 1) and 2), for all $i\in[k+1:K]$, $\zeta_i-\xi_i \le \frac{1}{2}\log\lp i+1\rp$. Hence,
\begin{align}
\delta_k 
&\le \sum_{i=k}^K\frac{1}{2}\log i + \sum_{i=k+1}^{K} \frac{1}{2}\log\lp i+1 \rp+\frac{1}{2}(K-k+1)\\
&\le \sum_{i=k}^K\frac{1}{2}\log K + \sum_{i=k+1}^{K-1}\frac{1}{2}\log K + \frac{1}{2}\log\lp K+1\rp+\frac{1}{2}(K-k+1)\\
&\le \lp K-k+1\rp\log K+\frac{1}{2}(K-k+1),
\end{align}
since $K+1 \le K^2$ for $K\ge2$.

\end{document}